\documentclass[pra,twocolumn]{revtex4}

\usepackage{amsmath}
\usepackage{amsthm}

\theoremstyle{remark}

\begin{document}

\newtheorem{prop}{Proposition}

\title{Measure-independent anomaly of nonlocality}
\author{S. Camalet}
\affiliation{Laboratoire de Physique Th\'eorique 
de la Mati\`ere Condens\'ee, UMR 7600, Sorbonne 
Universit\'es, UPMC Univ Paris 06, F-75005, 
Paris, France}

\begin{abstract}
We show that any Bell local state, with a hidden 
nonlocality that can be revealed by local 
filtering, is more, or equally, entangled than 
nonlocal states. More precisely, it can be 
deterministically transformed into a nonlocal 
state, by local operations and classical 
communication. For such a state, there is 
a clear anomaly of nonlocality, for 
any measures of entanglement and nonlocality. 
Moreover, we prove that the hidden nonlocality 
of any bipartite state more, or equally, 
entangled than nonlocal states, can be 
revealed by local operations and the sending 
of two one-bit messages, one in each 
direction. For some particular states, one bit 
of communication is even enough.
\end{abstract} 

\maketitle 

\section{Introduction}

Bell nonlocality and quantum entanglement are two 
distinct notions, whose relation is not straightforward 
\cite{W,Ba,BCPSS,PRA}. Entanglement is a quantum 
resource \cite{BHORS,BCP,GMNSH,BG,CG,aX} which 
cannot be generated by deterministic state 
transformations involving only local operations and 
classical communication (LOCC) \cite{DHR,HHHH}. 
In other words, a state is necessarily changed into 
a less, or equally, entangled state, by such a LOCC 
transformation. Thus, any proper measure of 
entanglement cannot increase under LOCC 
operations \cite{HHHH,BDSW}. The states with 
vanishing entanglement, are the separable states, 
which are the mixtures of product states \cite{W}. 
Any of them, can be reached by LOCC, from 
any state. 

Whereas entanglement is defined strictly within 
the framework of quantum mechanics, this is not 
the case for Bell nonlocality. The latter concerns 
the outcomes of local measurements, quantum or 
not, performed by distant observers. Their joint 
probabilities are said to be Bell local, if they can be 
reproduced by a hidden-variable model, in which 
a measurement outcome is determined only by 
the corresponding measurement, and the hidden 
variables, and does not depend on any other 
measurement \cite{BCPSS}. As well known, when 
a set of probabilities is Bell local, it satisfies Bell 
inequalities, such as 
the Clauser-Horne-Shimony-Holt (CHSH) 
inequality \cite{Be,CHSH}. In the opposite case, 
it is said nonlocal, and violates such an inequality. 

For a separable state, the joint probabilities of local 
measurement results, are always Bell local. But this is 
not specific to separable density operators. 
Some entangled states have this property \cite{Ba,W}. 
A state for which there are measurements violating 
a Bell inequality, is said to be nonlocal. A pure state 
is nonlocal if and only if it is entangled 
\cite{CFS,HS,G,PR}. Nevertheless, even 
for pure states, the relation between quantum 
entanglement and Bell nonlocality, is not obvious. 
For the simplest composite system, consisting of two 
two-level systems, the more entangled a pure state is, 
the more it can violate the CHSH inequality 
\cite{VeWo,MS}. But the situation is less clear for other 
Bell inequalities \cite{ADGL,LVB,VW,JP,R}, or other 
measures of nonlocality \cite{E,TB,AGG,BGS}, which are 
not maximum for maximally entangled states. 
However, it can be argued that they are not correct 
measures of nonlocality, and that, hence, no "anomaly 
of nonlocality" can be evidenced using them \cite{FP}.

Different procedures have been proposed to violate a 
Bell inequality with a Bell local entangled state, and 
hence, in some sense, reveal its hidden nonlocality. 
Such a violation can, for instance, be obtained by using 
several copies of the same state \cite{NV,Pa,CABV}, by 
performing local pre-measurements and selecting 
specific outcomes \cite{P,G2,HQBB}, by combining 
these two approaches \cite{Pe}, or by using more 
sophisticated techniques \cite{DSBBZ,CASA,KLMG}. 
In these scenarios, the density operator 
whose nonlocality is tested, is not the Bell local state 
of interest. The corresponding transformations change 
it into a state which is nonlocal, and possibly also 
more entangled. For example, several copies of 
a density operator, constitute a state more, or 
equally, entangled than this density operator, since 
the former is transformed into the latter, by local 
partial traces, which are LOCC operations. Thus, 
the above-mentioned procedures are not helpful in 
the understanding of a potential anomaly of nonlocality. 

In this paper, we show that there exist Bell local 
states which are more, or equally, entangled than 
nonlocal ones. In this case, the anomaly of 
nonlocality is manifest, since any nonlocality 
measure increases in going from such a Bell local 
state to a corresponding nonlocal one, whereas 
entanglement measures do not. The outline of 
the paper is as follows. We first recall, 
in Sec.\ref{Gd}, what are Bell nonlocality, LOCC 
operations, and local filtering operations. 
In Sec.\ref{NaL}, we prove our main result, namely, 
that a bipartite state which can be changed into 
a nonlocal one, by local filtering, is more, or equally, 
entangled than nonlocal states. In this section, we also 
discuss a particular state, which is Bell local, and as 
entangled as a nonlocal state. In Sec.\ref{Ccran}, 
we address the issue of the amount of classical 
communication, required to transform, by LOCC, 
a given bipartite state into a nonlocal state. We show 
that, whenever such a transformation is possible, 
two observers can achieve it with the transmission 
of only one bit from one observer to the other, and 
one bit in the opposite direction. Moreover, 
for some states, a single bit suffices. Finally, 
in Sec.\ref{C}, we summarize our results and 
discuss an important question they raise.

\section{Preliminaries}\label{Gd}

In this section, we introduce the notions used throughout 
the paper, namely, the local filtering 
\cite{P,HQBB,G2,BCPSS}, the LOCC-based ordering of 
quantum states \cite{DHR,HHHH}, and the Bell 
nonlocality \cite{Be,CHSH,BCPSS}.

\subsection{Local filtering and LOCC operations}

In this paper, we consider two kinds of quantum state 
transformations. For two systems, say A and B, 
whose Hilbert spaces are, respectively, 
${\cal H}_\mathrm{A}$ and ${\cal H}_\mathrm{B}$, 
a local filtering operation, described by the operators $M$ 
on ${\cal H}_\mathrm{A}$, and $N$ on 
${\cal H}_\mathrm{B}$, changes the state $\rho$, 
of A and B, into
\begin{equation}
\hat \rho=M\otimes N\rho M^\dag \otimes N^\dag/p , 
\label{f}
\end{equation}
where $p=\operatorname{tr} 
(M^\dag M\otimes N^\dag N \rho) \in (0,1]$ 
\cite{BCPSS}. The operators $M$ and $N$ are such 
that $M^\dag M \le I_\mathrm{A}$ 
and $N^\dag N \le I_\mathrm{B}$, where 
$I_\mathrm{A}$ and $I_\mathrm{B}$ are the identity 
operators of A and B, respectively. 
This filtering transformation is stochastic. It is achieved 
by performing measurements on A and B, and selecting 
specific outcomes. The state $\hat \rho$ is obtained 
with probability $p$. 

The other transformations of interest for our purpose, are 
the LOCC operations. A state $\rho$ is more, or equally, 
entangled than another one $\rho'$, if and only if there 
is a LOCC map $\Lambda$, such that 
$\rho'=\Lambda(\rho)$ \cite{HHHH}. Such 
a transformation is deterministic, i.e., it gives $\rho'$, 
from $\rho$, with probability unity. It is a composition of 
local partial traces, and of one-way LOCC operations 
of the form  
\begin{equation}
\Lambda(\rho)=
\sum_i ( F_i\otimes I_\mathrm{C} ) \rho 
( F_i \otimes I_\mathrm{C} )^\dag 
\otimes | i \rangle \langle i | ,
\label{owLOCC}
\end{equation}
where C is A or B, the linear maps 
$F_i:{\cal H}_\mathrm{D} 
\rightarrow {\cal H}_{\mathrm{E}}$, 
with D the other system, B or A, are such that 
$\sum_i F_i^\dag F_i^{\phantom{\dag}}
=I_\mathrm{D}$, 
and $| i \rangle$ are orthonormal states of an ancillary 
system, close to C, see Appendix. The system E can be D 
itself, a subsystem of it, or a system of which D is 
a subsystem. The transformation \eqref{owLOCC} involves 
a measurement on one of the systems, by an observer, 
and the sending, to another observer, of the outcome $i$, 
which is recorded using the ancilla. If $i$ has two possible 
values, only one bit is exchanged between the two 
observers. Operations of the form of eq.\eqref{owLOCC}, 
play an essential role in what follows.

\subsection{Bell nonlocality}

Let us denote $p(ij|kl)$ the probability of the outcomes 
$i$ and $j$, of measurements, indexed by $k$ and $l$, 
performed on systems A and B, respectively. The set 
$\{ p(ij|kl) \}_{i,j,k,l}$ is Bell local if and only if 
these probabilities can be written as
\begin{equation}
p(ij|kl)=\int d\lambda q(\lambda) 
p(i|k\lambda)p(j|l\lambda) , \label{Bl}
\end{equation}
where $\lambda$ denotes hidden variables, $q$ 
a probability density function, and $p(i|k\lambda)$ 
a probability distribution of the outcome $i$ of 
the measurement $k$ \cite{BCPSS}. 

Within the framework of quantum mechanics, 
a measurement, on A, with $m_k$ outcomes, is 
described by a set of positive operators, 
${\bf A}_k=\{ A_{i|k} \}_{i=1}^{m_k}$, 
such that $\sum_{i=1}^{m_k} A_{i|k}
=I_\mathrm{A}$, and joint probabilities 
of measurements on A and B, are given by
\begin{equation}
p(ij|kl)=\operatorname{tr} 
\left( \rho A_{i|k} \otimes B_{j|l} \right) , 
\label{qp}
\end{equation}
for the state $\rho$ of A and B, where $B_{j|l}$ are 
the operators describing the measurement $l$ on B 
\cite{J}. The state $\rho$ is nonlocal if there are 
measurements ${\bf A}_k$ and ${\bf B}_l$ such that 
$\{ p(ij|kl) \}_{i,j,k,l}$ does not satisfy eq.\eqref{Bl}. 
If, on the contrary, the probabilities \eqref{qp} can be 
written in the form of eq.\eqref{Bl}, for 
any measurements ${\bf A}_k$ and ${\bf B}_l$, $\rho$ 
is a Bell local state. 

\section{Revealing nonlocality by LOCC}\label{NaL}

Some Bell local states can be changed into nonlocal ones, 
by local filtering \cite{P,G2,HQBB}. For such a state 
$\rho$, we show below that there are nonlocal states less, 
or equally, entangled than $\rho$. We then consider a 
particular Bell local state, first studied in Ref.\cite{HQBB}, 
which is as entangled as a nonlocal state.

\subsection{Main result}

\begin{prop}\label{pf}
If a bipartite state can be changed into a nonlocal state, 
by local filtering, then two observers can deterministically 
transform it into a nonlocal state, using local operations, 
and the transmission of one bit from one observer to 
the other, and one bit in the opposite direction.
\end{prop} 
\begin{proof}
Consider a composite system AB, consisting of systems 
A and B, and a state $\rho$, of AB, such that there are 
operators $M$ and $N$ for which the state \eqref{f} is 
nonlocal. Since $M^\dag M \le I_\mathrm{A}$ 
and $N^\dag N \le I_\mathrm{B}$, there exist operators 
$\tilde M$ and $\tilde N$ such that 
$M^\dag M+\tilde M^\dag \tilde M=I_\mathrm{A}$, 
and $N^\dag N+\tilde N^\dag \tilde N=I_\mathrm{B}$. 
The LOCC operation given by eq.\eqref{owLOCC} 
with orthonormal states $|i \rangle$ of a two-level 
system $\mathrm{B'}$, $F_0=|0'\rangle M$, 
and $F_1=|1'\rangle \tilde M$, where $|i' \rangle$ are 
orthonormal states of a two-level system $\mathrm{A'}$, 
transforms $\rho$ into 
\begin{equation}
\rho_1= R_0 \otimes M \rho M^\dag \otimes S_0 
+ R_1 \otimes \tilde M \rho \tilde M^\dag 
\otimes S_1 , \label{1LOCC}
\end{equation}
where $R_i=|i' \rangle \langle i' |$, 
$S_i=|i \rangle \langle i |$, and the short-hand notation 
$M=M\otimes I_\mathrm{B}$ is used. A similar operation, 
with $N$ and $\tilde N$, in place of $M$ and $\tilde M$, 
respectively, and two more two-level systems, 
$\mathrm{A''}$ and $\mathrm{B''}$, transforms $\rho_1$ 
into
\begin{equation}
\rho_2= \sum_{i=0}^3 P_i \otimes K_i \rho K_i^\dag 
\otimes Q_i , \label{rho2}
\end{equation}
where $P_i$ ($Q_i$) are four projectors of system 
$\mathrm{A'A''}$ ($\mathrm{B'B''}$), 
summing to $I_\mathrm{A'A''}$ ($I_\mathrm{B'B''}$), 
$K_0=M\otimes N$, $K_1=M\otimes \tilde N$, 
$K_2=\tilde M\otimes N$, and 
$K_3=\tilde M\otimes \tilde N$.

The filtered state $\hat \rho$, given by eq.\eqref{f}, is 
nonlocal, by assumption. Thus, there are measurements 
${\bf A}_k$ and ${\bf B}_l$ such that the vector 
${\bf p}$, termed behavior \cite{T}, whose components 
are the probabilities $p(ij|kl)$, given by eq.\eqref{qp} 
with the density operator $\hat \rho$, does not satisfy 
eq.\eqref{Bl}. We define the positive operators, 
on ${\cal H}_\mathrm{A'A''A}$, 
\begin{equation}
\tilde A_{i|k}=P_0 \otimes A_{i|k} + 
\delta_{i,r_k} (I_\mathrm{A'A''}-P_0 ) \otimes 
I_\mathrm{A} , \label{meas}
\end{equation}
where $r_k \in \{ 1, \ldots, m_k \}$, with $m_k$ 
the number of outcomes of ${\bf A}_k$. The set 
${\bf \tilde A}_k=\{ \tilde A_{i|k} \}_{i=1}^{m_k}$ 
constitutes a measurement, on $\mathrm{A'A''A}$, since 
$\sum_{i=1}^{m_k} \tilde A_{i|k}=I_\mathrm{A'A''A}$. 
Similarly, from a measurement ${\bf  B}_l$, on B, with 
$n_l$ outcomes, a measurement ${\bf \tilde B}_l$, 
on $\mathrm{BB'B''}$, involving an integer 
$s_l \in \{ 1, \ldots, n_l \}$, can be defined. 
The probabilities $\tilde p(ij|kl)=\operatorname{tr} 
( \rho_2 \tilde A_{i|k} \otimes \tilde B_{j|l})$ can be 
written as
\begin{equation}
\tilde p(ij|kl) = p p(ij|kl) 
+ (1-p) d_\lambda(ij|kl) , \label{rp}
\end{equation} 
where 
$p=\operatorname{tr} 
(M^\dag M\otimes N^\dag N \rho)$, 
$d_\lambda(ij|kl)=\delta_{i,r_k}\delta_{j,s_l}$, 
and $\lambda=(r_1,r_2,\ldots,s_1,s_2,\ldots)$.

The behavior ${\bf \tilde p}$ is Bell local if and only if 
it belongs to the compact convex polytope
\begin{equation}
{\cal L}=\left\{ \sum_\lambda q_\lambda {\bf d}_\lambda 
: q_ \lambda \ge 0, \sum_\lambda q_\lambda = 1 \right\} , 
\label{L}
\end{equation}
where the sums run over all $\lambda$ \cite{F,BCPSS}. 
There is a particular $\lambda$ such that 
${\bf \tilde p}=p{\bf p}+\bar p{\bf d}_\lambda 
\notin {\cal L}$, where $\bar p=1-p$. This can be seen 
as follows. Assume that, for any $\lambda$, 
$p{\bf p}+\bar p{\bf d}_\lambda \in {\cal L}$. 
This implies, together with the convexity of ${\cal L}$, 
that $p{\bf p}+\bar p{\bf p}' \in {\cal L}$, for any Bell 
local behavior 
${\bf p}'=\sum_\lambda q_\lambda {\bf d}_\lambda$. 
This gives a sequence of elements of ${\cal L}$, 
$[1-\bar p^n]{\bf p}+\bar p^n {\bf p}'$, that converges 
to ${\bf p}$, which is not possible since ${\cal L}$ is 
closed, and ${\bf p} \notin {\cal L}$. In conclusion, 
there are measurements ${\bf \tilde A}_k$ and 
${\bf \tilde B}_l$ for which the corresponding behavior 
${\bf \tilde p}$ is not in ${\cal L}$, and hence $\rho_2$ 
is nonlocal. The two-stage LOCC transformation 
$\rho \mapsto \rho_1 \mapsto \rho_2$ involves 
the sending of two one-bit messages, one in each 
direction. 
\end{proof}

The exchange and storage of classical information play 
a crucial role in the above transformation leading to 
a nonlocal state. To see it, consider the state $\rho_3$ 
obtained from the state \eqref{rho2}, by tracing out 
the systems $\mathrm{B'}$ and $\mathrm{A''}$, used 
to record the bits exchanged between the two observers. 
This density operator can be written as 
$\rho_3=\Lambda_\mathrm{B} 
\circ \Lambda_\mathrm{A} (\rho)$, where 
$\Lambda_\mathrm{A}$ and $\Lambda_\mathrm{B}$ are 
local operations. Thus, it is Bell local if $\rho$ is \cite{Ba}.

\subsection{Example}

As an example, we consider the state, of two three-level 
systems, A and B,
\begin{equation}
\rho=p|\psi \rangle \langle \psi |+p M\otimes \tilde N
+q\tilde M\otimes N+4q\tilde M\otimes \tilde N , 
\label{rho}
\end{equation}
where 
$|\psi \rangle=(|0 \rangle|0' \rangle+|1 \rangle|1' \rangle)
/\sqrt{2}$, with orthonormal states $|i \rangle$ of A, 
and $|i' \rangle$ of B, $q=(1-3p)/6$, $p \le 1/18$, 
$M=|0 \rangle \langle 0 |+|1 \rangle \langle 1 |$, 
and $\tilde M=|2 \rangle \langle 2 |$. The operators $N$ 
and $\tilde N$, on ${\cal H}_{\mathrm{B}}$, are given 
by similar expressions. This state has been shown to be 
Bell local \cite{HQBB}. The corresponding filtered state 
\eqref{f} is $|\psi \rangle \langle \psi |$, which maximally 
violates the CHSH inequality. It is attained with 
probability $p$.

The state \eqref{rho2}, obtained by LOCC from $\rho$, is 
here
\begin{multline}
\rho_2= p P_0 \otimes |\psi \rangle \langle \psi |
\otimes Q_0 
+ p P_1\otimes M\otimes \tilde N \otimes Q_1 \\
+q P_2\otimes \tilde M \otimes N \otimes Q_2
+4q P_3\otimes \tilde M \otimes \tilde N \otimes Q_3 , 
\label{rho2exp}
\end{multline}
It results, from the above proof, that it is nonlocal. 
This can be shown directly as follows. We find 
\begin{equation}
\left\langle A_1 (B_1+B_2) + A_2 (B_1-B_2) \right\rangle 
= 2 p (\sqrt{2}-1) + 2 \ge 2 , \nonumber
\end{equation}
where $A_k B_l=A_k \otimes B_l$, and 
$\langle \ldots \rangle=\operatorname{tr}(\rho_2 \ldots)$, 
for the dichotomic observables 
\begin{eqnarray}
A_k&=&P_0 \otimes \big(\sigma_k+\tilde M\big) 
+ \tilde P\otimes I_\mathrm{A} , \nonumber \\
B_l&=&\big([\sigma'_1+(3-2l) \sigma'_2]/\sqrt{2} 
+\tilde N\big)\otimes Q_0
+I_\mathrm{B} \otimes \tilde Q , \nonumber
\end{eqnarray}
where 
$\sigma_1=|0 \rangle \langle 0 |-|1 \rangle \langle 1 |$, 
$\sigma_2=|0 \rangle \langle 1 |+|1 \rangle \langle 0 |$, 
$\sigma'_1$ and $\sigma'_2$ are defined similarly 
for system B, $\tilde P=I_\mathrm{A'A''}-P_0$, 
and $\tilde Q=I_\mathrm{B'B''}-Q_0$. That is 
to say, $\rho_2$ violates the CHSH inequality. Other 
measurements may lead to a larger violation. It is also 
possible that other Bell inequalities could be more 
appropriate. The local operations consisting in tracing 
out the systems $\mathrm{A'A''}$ and $\mathrm{B'B''}$, 
change the state \eqref{rho2exp} back into the state 
\eqref{rho}. They are thus equally entangled, whereas 
one is Bell local and the other is nonlocal.

\section{Amount of communication required to reveal 
nonlocality}\label{Ccran}

We have seen above that the hidden nonlocality of some 
bipartite states, can be revealed by sending two one-bit 
messages, one in each direction. One can wonder whether 
this can be achieved with less communication. Some 
exchange of information is necessary, since local 
operations alone cannot transform a Bell local state into 
a nonlocal one \cite{Ba}. We show below that one bit 
of communication is enough for some states. We then 
prove that, for any state more, or equally, entangled than 
nonlocal states, it can be done with only one bit per 
direction.

\subsection{One bit is the minimum}

The proof below applies to Bell local states that can be 
changed into nonlocal states, with only one local filter. 
Let us first show that there exist such states. Consider 
a Bell local state $\rho$, such that the filtered state 
\eqref{f} is nonlocal, and the state 
$\rho' \propto M \otimes I_\mathrm{B}\rho 
M^\dag \otimes I_\mathrm{B}$, obtained, from $\rho$, 
by applying only the filter described by $M$. The latter is 
transformed into the nonlocal state \eqref{f}, by the filter 
described by $N$. If $\rho'$ is nonlocal, then the hidden 
nonlocality of $\rho$ can be revealed with one local filter. 
If, on the contrary, $\rho'$ is Bell local, then its hidden 
nonlocality can be revealed with one local filter.
\begin{prop}\label{p1f}
If a bipartite state can be changed into a nonlocal state, 
with one local filter, then it can be deterministically 
transformed into a nonlocal state, with local operations, 
and one bit of communication.
\end{prop} 
\begin{proof}
Consider two systems A and B, and a state $\rho$ 
of these systems, such that there is an operator $M$ 
on ${\cal H}_A$, for which the state 
$\rho'=M\otimes I_\mathrm{B}\rho 
M^\dag\otimes I_\mathrm{B}/p$, where 
$p=\operatorname{tr} 
(M^\dag M\otimes I_\mathrm{B}\rho)$, 
is nonlocal. There exists $\tilde M$ such that 
$M^\dag M+\tilde M^\dag \tilde M=I_\mathrm{A}$. 
Since $\rho'$ is nonlocal, there are measurements 
${\bf A}_k$, with $m_k$ outcomes, and ${\bf B}_l$, 
with $n_l$ outcomes, such that the behavior ${\bf p}$, 
whose components are the probabilities $p(ij|kl)$, 
given by eq.\eqref{qp} with the density operator $\rho'$, 
is nonlocal, i.e., does not satisfy eq.\eqref{Bl}. Let us 
introduce two two-level systems $\mathrm{A'}$ and 
$\mathrm{B'}$, and define the positive operators, 
on ${\cal H}_\mathrm{A'A}$, 
\begin{equation}
\tilde A_{i|k}=|0' \rangle \langle 0' | \otimes A_{i|k} + 
\delta_{i,r_k} |1' \rangle \langle 1' | \otimes I_\mathrm{A} 
, \nonumber
\end{equation} 
where $|i' \rangle$ are orthonormal states of 
$\mathrm{A'}$, and $r_k \in \{ 1, \ldots, m_k \}$, 
and similar ones, $\tilde B_{j|l}$, 
on ${\cal H}_\mathrm{BB'}$, with an integer 
of $\{ 1, \ldots, n_l \}$, and orthonormal states 
of $\mathrm{B'}$. The set 
${\bf \tilde A}_k=\{ \tilde A_{i|k} \}_{i=1}^{m_k}$ 
constitutes a measurement on $\mathrm{A'A}$, 
and ${\bf \tilde B}_l$ on $\mathrm{BB'}$. 
The probabilities 
$\tilde p(ij|kl)=\operatorname{tr} ( \rho_1 \tilde A_{i|k} 
\otimes \tilde B_{j|l})$, where $\rho_1$ is given 
by eq.\eqref{1LOCC}, can be cast into the form of 
eq.\eqref{rp}, with $p$ defined above. Since ${\bf p}$ 
is nonlocal, there is a deterministic behavior 
${\bf d}_\lambda$, such that 
${\bf \tilde p}=p{\bf p}+(1-p){\bf d}_\lambda$ is 
nonlocal, see the proof of proposition \ref{pf}, and 
hence $\rho_1$ is nonlocal. The LOCC transformation 
$\rho \mapsto \rho_1$ involves the sending of only one 
one-bit message. 
\end{proof}

\subsection{One bit per direction is enough}

\begin{prop}
A bipartite state is more, or equally, entangled than 
a nonlocal state, if and only if it can be deterministically 
transformed into a nonlocal state, with local operations, 
and two one-bit messages, one in each direction.
\end{prop}
\begin{proof}
Consider a state $\rho$, of systems A and B, such that 
there is a LOCC operation $\Lambda$, for which 
$\rho'=\Lambda(\rho)$ is nonlocal. The map $\Lambda$ 
is separable \cite{DHR,BDF}, i.e., 
$\rho'=\sum_i q_i \omega_i$ with the states 
$\omega_i=  
M_i\otimes N_i \rho M_i^\dag \otimes N_i^\dag  /q_i$, 
where the operators $M_i$ and $N_i$, 
on ${\cal H}_\mathrm{A}$ and ${\cal H}_\mathrm{B}$, 
respectively, are such that 
$\sum_i M_i^\dag M_i^{\phantom{\dag}}  
\otimes N_i^\dag N_i^{\phantom{\dag}}
=I_\mathrm{AB}$, 
and $q_i=\operatorname{tr} 
(M_i^\dag M_i^{\phantom{\dag}} 
\otimes N_i^\dag N_i^{\phantom{\dag}}  \rho)$. 
Since $\rho'$ is nonlocal, there are measurements 
${\bf A}_k$ and ${\bf B}_l$ such that 
the corresponding behavior ${\bf p}$ is nonlocal. 
It can be written as ${\bf p}=\sum_i q_i {\bf p}_i$, 
where ${\bf p}_i$ is the behavior for the state $\omega_i$, 
and the measurements ${\bf A}_k$ and ${\bf B}_l$. 
Since the set ${\cal L} $ of the Bell local behaviors, given 
by eq.\eqref{L}, is convex, ${\bf p} \notin {\cal L}$, 
and $\sum_i q_i=1$, there is $\iota$ such that 
${\bf p}_{\iota} \notin {\cal L}$. 

The corresponding operators $M_{\iota}$ and 
$N_{\iota}$ obey 
$M_{\iota}^\dag M^{\phantom{\dag}}_{\iota} 
\le I_\mathrm{A}$, and 
$N_{\iota}^\dag N^{\phantom{\dag}}_{\iota} 
\le I_\mathrm{B}$, see Appendix. Thus, 
there are $\tilde M$ and $\tilde N$ such that 
$M_{\iota}^\dag M^{\phantom{\dag}}_{\iota}
+\tilde M^\dag \tilde M=I_\mathrm{A}$, and 
$N_{\iota}^\dag N^{\phantom{\dag}}_{\iota}
+\tilde N^\dag \tilde N=I_\mathrm{B}$. As shown 
in the proof of proposition \ref{pf}, there is 
a LOCC operation, involving the sending 
of two one-bit messages, one in each direction, that 
transforms $\rho$ into the state $\rho_2$, given by 
eq.\eqref{rho2} with 
$K_0=M_{\iota}\otimes N_{\iota}$, 
$K_1=M_{\iota}\otimes \tilde N$, 
$K_2=\tilde M\otimes N_{\iota}$, and 
$K_3=\tilde M\otimes \tilde N$. 
Since ${\bf p}_{\iota}\notin {\cal L}$, there are 
measurements ${\bf \tilde A}_k$, given by 
eq.\eqref{meas}, and ${\bf \tilde B}_l$, defined similarly 
from ${\bf B}_l$, such that the behavior of components 
$\operatorname{tr} 
( \rho_2 \tilde A_{i|k} \otimes \tilde B_{j|l})$, is nonlocal, 
and hence $\rho_2$ is nonlocal, see the proof of 
proposition \ref{pf}. 

The converse follows directly from the definition of 
the entanglement ordering.
\end{proof}

\section{Conclusion}\label{C}

In summary, we have shown that there are states which 
are Bell local, but more, or equally, entangled than 
nonlocal ones. They are those with a hidden 
nonlocality that can be revealed by local filtering. 
For these states, there is a clear anomaly of nonlocality, 
for any measures of entanglement and nonlocality. 
We have also proved that any state more, or equally, 
entangled than nonlocal ones, can be changed into 
a nonlocal state, with local operations and only two 
bits of communication, one in each direction. 
For some particular states, a single bit is even enough. 

A natural question arising from these results, is whether 
all entangled states are more, or equally, entangled than 
nonlocal states. In other words, do all Bell local 
entangled states have a hidden nonlocality that can be 
revealed by LOCC ? It has been shown recently that 
the answer to the similar question for local filtering, is 
negative \cite{HQBVB}. But this does not imply 
a negative answer for LOCC. Due to our last result, 
this issue can be addressed by considering only LOCC 
operations involving the sending of a single one-bit 
message per direction.

\section*{APPENDIX: ONE-WAY LOCC DECOMPOSITION}

In this Appendix, we show how any LOCC operation can 
be obtained from a sequence of one-way LOCC maps of 
the form of eq.\eqref{owLOCC}. 
Any LOCC transformation of a state $\rho$ of a bipartite 
system $\mathrm{A}_1\mathrm{B}_1$, can be written as 
$\Lambda(\rho)=\sum_{\bf i} 
K^{\phantom{\dag}}_{\bf i} \rho K^\dag_{\bf i}$, 
where ${\bf i}=(i_1, \ldots, i_{2n})$, $i_r$ runs from 
$1$ to $d_r$, and 
\begin{eqnarray}
K_{\bf i}&=&
\left(M^{(2n-1)}_{{\bf i}_{2n-1}}
\otimes N^{(2n)}_{{\bf i}_{2n}}\right)
\ldots \left(M^{(1)}_{i_1} 
\otimes N^{(2)}_{{\bf i}_2}\right) , \nonumber \\ 
\nonumber
&=&M^{(2n-1)}_{{\bf i}_{2n-1}} \ldots M^{(1)}_{i_1}
\otimes N^{(2n)}_{{\bf i}_{2n}}\ldots 
N^{(2)}_{{\bf i}_2} ,
\end{eqnarray} 
with ${\bf i}_r=(i_1, \ldots, i_r)$. The linear maps 
$M^{(2r-1)}_{{\bf i}_{2r-1}}:
{\cal H}_{\mathrm{A}_r}\rightarrow 
{\cal H}_{\mathrm{A}_{r+1}}$ satisfy
\begin{equation}
\sum_{i_{2r-1}=1}^{d_{2r-1}} 
\left( M^{(2r-1)}_{{\bf i}_{2r-2},i_{2r-1}}\right)^\dag
M^{(2r-1)}_{{\bf i}_{2r-2},i_{2r-1}}=I_{\mathrm{A}_r} 
\nonumber ,
\end{equation} 
and the operators $N^{(2r)}_{{\bf i}_{2r}}:
{\cal H}_{\mathrm{B}_r}\rightarrow 
{\cal H}_{\mathrm{B}_{r+1}}$ obey similar 
relations \cite{DHR}. We remark that the above 
equality gives, for any $| \psi \rangle 
\in {\cal H}_{\mathrm{A}_r}$,
$\langle \psi |
( M^{(2r-1)}_{{\bf i}_{2r-1}})^\dag
M^{(2r-1)}_{{\bf i}_{2r-1}}| \psi \rangle 
\le \langle \psi | \psi \rangle$, and hence, 
for any $| \psi \rangle 
\in {\cal H}_{\mathrm{A}_1}$,

\begin{equation}
\langle \psi |\left( 
M^{(2n-1)}_{{\bf i}_{2n-1}} \ldots M^{(1)}_{i_1} 
\right)^\dag
M^{(2n-1)}_{{\bf i}_{2n-1}} \ldots M^{(1)}_{i_1}
| \psi \rangle \le \langle \psi | \psi \rangle .
\nonumber
\end{equation}

Let us introduce the systems 
$\mathrm{A}'_r$ and $\mathrm{B}'_r$, 
of Hilbert space dimension $d_r$, 
where $r \in \{1, \ldots,2n\}$, 
and the composite systems $\mathrm{A}^{[r]}
=\mathrm{A}'_1 
\ldots \mathrm{A}'_{2r-1}\mathrm{A}_{r+1}$, 
and $\mathrm{B}^{[r]}
=\mathrm{B}_{r+1}\mathrm{B}'_1 
\ldots \mathrm{B}'_{2r}$. 
We denote by $P^{(r)}_{i_r}$ projectors such that 
$\sum_{i_r=1}^{d_r} P^{(r)}_{i_r}
=I_{\mathrm{A}'_r}$, 
and $Q^{(r)}_{i_r}$ similar projectors for 
$\mathrm{B}'_r$, and define the one-way LOCC 
operations $\Lambda_1,\ldots, \Lambda_{2n}$ by
\begin{widetext}
\begin{eqnarray}
\Lambda_1(\rho_1)&=&\sum_{i_1} 
P^{(1)}_{i_1} \otimes 
\left(M^{(1)}_{i_1}\otimes I_{\mathrm{B}_1} \right)  
\rho_1 
\left(M^{(1)}_{i_1}\otimes I_{\mathrm{B}_1} 
\right)^\dag 
\otimes Q^{(1)}_{i_1} \nonumber \\
\Lambda_{2r}(\rho_{2r})&=&\sum_{{\bf i}_{2r}} 
P^{(2r)}_{i_{2r}} \otimes
\left( I_{\mathrm{A}^{[r]}}
\otimes N^{(2r)}_{{\bf i}_{2r}} 
\otimes Q_{{\bf i}_{2r-1}} \right) \rho_{2r} 
\left( I_{\mathrm{A}^{[r]}}
\otimes N^{(2r)}_{{\bf i}_{2r}} 
\otimes Q_{{\bf i}_{2r-1}}  \right)^\dag 
\otimes Q^{(2r)}_{i_{2r}} \nonumber \\
\Lambda_{2r+1}(\rho_{2r+1})
&=&\sum_{{\bf i}_{2 r+1}} 
P^{(2 r+1)}_{i_{2 r+1}} \otimes
\left( P_{{\bf i}_{2r}} 
\otimes M^{(2 r+1)}_{{\bf i}_{2 r+1}} 
\otimes I_{\mathrm{B}^{[r]}} \right)\rho_{2r+1} 
\left(  P_{{\bf i}_{2r}} 
\otimes M^{(2 r+1)}_{{\bf i}_{2 r+1}} 
\otimes I_{\mathrm{B}^{[r]}}  \right)^\dag 
\otimes Q^{(2 r+1)}_{i_{2 r+1}} , \nonumber 
\end{eqnarray}
\end{widetext}
where $P_{{\bf i}_r}=P^{(r)}_{i_r} \otimes 
P^{(r-1)}_{i_{r-1}} \otimes \ldots 
\otimes P^{(1)}_{i_1}$, $i_r$ runs from 1 to 
$d_r$, $\rho_1$ 
is a state of $\mathrm{A}_1\mathrm{B}_1$, 
$\rho_{2r}$ of $\mathrm{A}^{[r]}\mathrm{B}_{r}
\mathrm{B}'_1 \ldots \mathrm{B}'_{2r-1}$, 
and $\rho_{2r+1}$ of 
$\mathrm{A}'_1 \ldots \mathrm{A}'_{2r}
\mathrm{A}_{r+1}\mathrm{B}^{[r]}$. 
The LOCC map 
$\Phi=\Lambda_{2n} \circ \ldots \circ \Lambda_1$ 
transforms a state $\rho$ of 
$\mathrm{A}_1\mathrm{B}_1$, 
into $\Phi(\rho)=\sum_{\bf i} P_{\bf i} \otimes
K^{\phantom{\dag}}_{\bf i} \rho K^\dag_{\bf i} 
\otimes Q_{\bf i}$. Tracing out the ancillary systems 
$\mathrm{A}'_r$ and $\mathrm{B}'_r$, 
gives $\Lambda(\rho)$.


\begin{thebibliography}{99}

\bibitem{W} R.F. Werner, Quantum states with 
Einstein-Podolsky-Rosen correlations admitting a 
hidden-variable model, Phys. Rev. A {\bf 40}, 
4277 (1989).

\bibitem{Ba} J. Barrett, Nonsequential 
positive-operator-valued measurements on entangled 
mixed states do not always violate a Bell inequality, 
Phys. Rev. A {\bf 65}, 042302 (2002).

\bibitem{BCPSS} N. Brunner, D.Cavalcanti, S. Pironio, 
V. Scarani, and S. Wehner, Bell nonlocality, Rev. Mod. 
Phys. {\bf 86}, 419 (2014).

\bibitem{PRA} S.Camalet, Monogamy inequality 
for entanglement and local contextuality, 
Phys. Rev. A {\bf 95}, 062329 (2017).

\bibitem{BHORS} F. G. S. L. Brand\~ ao, 
M. Horodecki, J. Oppenheim, J. M. Renes, and 
R. W. Spekkens, Resource Theory of Quantum 
States Out of Thermal Equilibrium, 
Phys. Rev. Lett. {\bf 111}, 250404 (2013).

\bibitem{BCP} T. Baumgratz, M. Cramer, 
and M.B. Plenio, Quantifying Coherence, 
Phys. Rev. Lett. {\bf 113}, 140401 (2014).

\bibitem{GMNSH} G. Gour, M. P. M\" uller, 
V. Narasimhachar, R. W. Spekkens, and 
N. Y. Halpern, The resource theory of 
informational nonequilibrium in thermodynamics, 
Phys. Rep. {\bf 583}, 1 (2015).

\bibitem{BG} F. G.S.L Brand\~ ao, and G. Gour, 
Reversible Framework for Quantum Resource 
Theories, Phys. Rev. Lett. {\bf 115}, 070503 (2015).

\bibitem{CG} E. Chitambar, and G. Gour, 
Critical Examination of Incoherent Operations 
and a Physically Consistent Resource Theory 
of Quantum Coherence, Phys. Rev. Lett. 
{\bf 117}, 030401 (2016).

\bibitem{aX} S.Camalet, Monogamy inequality 
for any local quantum resource and entanglement, 
arXiv:1704.03199, to be published in Phys. Rev. Lett.

\bibitem{DHR} M. J. Donald, M. Horodecki, and 
O. Rudolph, The uniqueness theorem for entanglement 
measures, J.  Math. Phys. {\bf 43}, 4252 (2002).

\bibitem{HHHH} R. Horodecki, P. Horodecki, M. Horodecki, 
and K. Horodecki, Quantum entanglement, Rev. Mod. 
Phys. {\bf 81}, 865 (2009).

\bibitem{BDSW} C.H. Bennett, D. P. DiVincenzo, J. Smolin, 
and W. K. Wootters, Mixed-state entanglement and 
quantum error correction, Phys. Rev. A {\bf 54}, 3824 
(1996).

\bibitem{Be} J.S. Bell, On the Einstein Podolsky Rosen 
paradox, Physics (Long Island City, N.Y.) {\bf 1}, 195 
(1964).

\bibitem{CHSH} J.F. Clauser, M.A. Horne, A. Shimony, 
and R.A. Holt, Proposed experiment to test local 
hidden-variable theories, Phys. Rev. Lett. {\bf 23}, 880 
(1969).

\bibitem{CFS} V. Capasso, D. Fortunato, and F. Selleri, 
Int. J. Mod. Phys. {\bf 7}, 319 (1973).

\bibitem{HS} D. Home, and F. Selleri, Bell's theorem and 
the EPR paradox, Riv. Nuovo Cim. {\bf 14}, 1 (1991).

\bibitem{G} N. Gisin, Bell's inequality holds for all 
non-product states, Phys. Lett. A {\bf 154}, 201 (1991).

\bibitem{PR} S. Popescu and D. Rohrlich, Generic quantum 
nonlocality, Phys. Lett. A {\bf 166}, 293 (1992).

\bibitem{VeWo} F. Verstraete and M.M. Wolf, 
Entanglement versus Bell Violations and Their Behavior 
under Local Filtering Operations, Phys. Rev. Lett. {\bf 89}, 
170401 (2002).

\bibitem{MS} A.A. M\'ethot, and V. Scarani, An anomaly 
of non-locality, Quant. Inf. Comp. {\bf 7}, 157 (2007).

\bibitem{ADGL} A. Ac\' in, T. Durt, N. Gisin, and 
J.I. Latorre, Quantum nonlocality in two three-level 
systems, Phys. Rev. A {\bf 65}, 052325 (2002).

\bibitem{LVB} Y.-C. Liang, T. V\'ertesi, and N. Brunner, 
Semi-device-independent bounds on entanglement, 
Phys. Rev. A {\bf 83}, 022108 (2011).

\bibitem{VW} T. Vidick, and S. Wehner, More nonlocality 
with less entanglement, Phys. Rev. A {\bf 83}, 052310 
(2011).

\bibitem{JP} M. Junge, and C. Palazuelos, Large violation 
of Bell inequalities with low entanglement, 
Comm. Math. Phys. {\bf 306}, 695 (2011).

\bibitem{R} O. Regev, Bell violations through independent 
bases games, Quant. Inf. Comp. {\bf 12}, 9 (2012).

\bibitem{E} P.H. Eberhard, Background level and counter 
efficiencies required for a loophole-free 
Einstein-Podolsky-Rosen experiment, 
Phys. Rev. A {\bf 47}, R747 (1993).

\bibitem{TB} B. F. Toner and D. Bacon, Communication 
cost of simulating Bell correlations, 
Phys. Rev. Lett. {\bf 91}, 187904 (2003).

\bibitem{AGG} A. Ac\' in, R. Gill, and N. Gisin, Optimal 
Bell tests do not require maximally entangled states, 
Phys. Rev. Lett. {\bf 95}, 210402 (2005).

\bibitem{BGS} N. Brunner, N. Gisin, and V. Scarani, 
Entanglement and non-locality are different resources, 
New J. Phys {\bf 7}, 88 (2005).

\bibitem{FP} E. A. Fonseca, and F. Parisio, Measure 
of nonlocality which is maximal for maximally entangled 
qutrits, Phys. Rev. A {\bf 92}, 030101(R) (2015).

\bibitem{NV} M. Navascu\'es, and T. V\'ertesi, Activation 
of nonlocal quantum resources, Phys. Rev. Lett. {\bf 106},
060403 (2011).

\bibitem{Pa} C. Palazuelos, Superactivation of quantum 
nonlocality, Phys. Rev. Lett. {\bf 109}, 190401 (2012).

\bibitem{CABV} D. Cavalcanti, A. Ac\'in, N. Brunner, and 
T. Vert\'esi, All quantum states useful for teleportation 
are nonlocal resources, Phys. Rev. A {\bf 87}, 042104 
(2013).

\bibitem{P} S. Popescu, Bell's Inequalities and Density 
Matrices: Revealing "Hidden" Nonlocality, 
Phys. Rev. Lett. {\bf 74}, 2619 (1995).

\bibitem{G2} N. Gisin, Hidden quantum nonlocality 
revealed by local filters, Phys. Lett. A {\bf 210}, 151 
(1996).

\bibitem{HQBB} F. Hirsch, M. T. Quintino, J. Bowles, 
and N. Brunner, Genuine hidden quantum nonlocality, 
Phys. Rev. Lett. {\bf 111}, 160402 (2013).

\bibitem{Pe} A. Peres, Collective tests for quantum 
nonlocality, Phys. Rev. A {\bf 54}, 2685 (1996).

\bibitem{DSBBZ} A. Sen(De), U. Sen, C. Brukner, 
V. Bu\v zek, and M. \.Zukowski, Entanglement swapping 
of noisy states: A kind of superadditivity in nonclassicality, 
Phys. Rev. A {\bf 72}, 042310 (2005).

\bibitem{CASA} D. Cavalcanti, M. L. Almeida, V. Scarani, 
and A. Ac\'in, Quantum networks reveal quantum 
nonlocality, Nat. Commun. {\bf 2}, 184 (2011).

\bibitem{KLMG} W. Klobus, W. Laskowski, M. Markiewicz, 
and A. Grudka, Nonlocality activation in 
entanglement-swapping chains, Phys. Rev. A {\bf 86}, 
020302(R) (2012).

\bibitem{J} K. Jacobs, {\it Quantum measurement 
theory and its applications} (Cambridge University Press, 
Cambridge, 2014).

\bibitem{T} B.S. Tsirelson, Some results and problems 
on quantum Bell-type inequalities, 
Hadronic J. Suppl. {\bf 8}, 329 (1993).

\bibitem{F} A. Fine, Hidden variables, joint probability, 
and the Bell inequalities, Phys. Rev. Lett. {\bf 48}, 291 
(1982).

\bibitem{BDF} C.H. Bennett, D.P. DiVincenzo, 
C.A. Fuchs, T. Mor, E. Rains, P.W. Shor, J.A. Smolin 
and W.K. Wootters, Quantum nonlocality without 
entanglement, Phys. Rev. A {\bf 59}, 1070 (1999).

\bibitem{HQBVB} F. Hirsch, M. T. Quintino, J. Bowles, 
T. V\' ertesi, and N. Brunner, Entanglement without 
hidden nonlocality, New J. Phys. {\bf 18}, 113019 
(2016).

\end{thebibliography}
\end{document}